\documentclass[PCJ,Unicode,screen]{cedram}

\usepackage{setspace}
\providecommand*{\noopsort}[1]{}

\usepackage{appendix}
\usepackage{amsmath,amssymb,amsthm,bbm}
\usepackage{hyperref}
\usepackage{cleveref} 
\hypersetup{colorlinks=true}
\usepackage{caption}
\usepackage{bbold}
\usepackage{thmtools}
\usepackage{mathtools}
\usepackage{enumerate}

\newtheoremstyle{definitionbf}%
  {}{}{}{}{\bfseries}{.}{ }%
  {\thmname{#1}\thmnumber{ #2}\thmnote{ (#3)}}
\newtheoremstyle{plainbf}%
  {}{}{\itshape}{}{\bfseries}{.}{ }%
  {\thmname{#1}\thmnumber{ #2}\thmnote{ (#3)}}
\newtheoremstyle{bfnoteonly}%
{}{}%
{\itshape}{}%
{\bfseries}{.}%
{ }%
{\thmnote{#3}}

\usepackage{natbib} 

\theoremstyle{plainbf}
\newtheorem{theorem}{Theorem} 
\newtheorem{lemma}[theorem]{Lemma}
\newtheorem{proposition}[theorem]{Proposition}

\theoremstyle{definitionbf} 
\newtheorem{definition}{Definition} 

\theoremstyle{remark}

\theoremstyle{bfnoteonly}
\newtheorem*{model}{Model}

\newcommand{\com}[1]{} 

\newcommand{\RR}{\mathbb{R}}

\newcommand{\EE}{\mathbb{E}}

\newcommand{\indi}{\mathbb{1}}

\DeclareMathOperator{\im}{Im}
\DeclareMathOperator{\var}{var}

\newcommand{\norm}[1]{\left\lVert#1\right\rVert}

\newcommand{\Ddisp}{D_\mathrm{disp}}
\newcommand{\Qavg}{Q_\mathrm{avg}}
\newcommand{\Qdti}{Q_\mathrm{DTI}}
\newcommand{\Vdti}{V_\mathrm{DTI}}
\newcommand{\Pidti}{\Pi_\mathrm{DTI}}

\newcommand{\ha}{\hat{\alpha}}
\newcommand{\hb}{\hat{b}}
\newcommand{\htt}{\hat{t}}
\newcommand{\as}{\alpha^*}
\newcommand{\bs}{b^*}
\newcommand{\ts}{t^*}
\newcommand{\bes}{\beta^*}
\newcommand{\sk}{^{(k)}}

\usepackage{xcolor}
\definecolor{changecolor}{RGB}{192,64,0}

\usepackage[normalem]{ulem} 

\renewcommand{\printDOI}{}

\PCIcorresp{cecile.ane@wisc.edu}

\title{
A consistent least-squares criterion for calibrating
edge lengths in phylogenetic networks}
\date{}
 \author{\firstname{Jingcheng} \lastname{Xu}}
 \address{Department of Statistics, University of Wisconsin - Madison, USA.}
 \email{xjc@stat.wisc.edu}
 \author{\firstname{C\'ecile} \lastname{An\'e}}
 \address{Department of Statistics, University of Wisconsin - Madison, USA.}
 \address{Department of Botany, University of Wisconsin - Madison, USA.}
 \email{cecile.ane@wisc.edu}


\begin{abstract}
  In phylogenetic networks, it is desirable to estimate edge lengths in
  substitutions per site or calendar time. Yet,
  there is a lack of scalable methods that provide such estimates.  Here we
  consider the problem of obtaining edge length estimates from genetic
  distances, in the presence of rate variation across genes and lineages, when
  the network topology is known.  We propose a novel criterion based on
  least-squares that is both consistent and computationally tractable.  The crux
  of our approach is to decompose the genetic distances into two parts, one of
  which is invariant across displayed trees of the network.  The scaled genetic
  distances are then fitted to the invariant part, while the average scaled
  genetic distances are fitted to the non-invariant part.  We show that this
  criterion is consistent provided that there exists a tree path
  between some pair of tips in the network,
  and that edge lengths in the network are identifiable from average distances.
  We also provide a constrained variant of this criterion assuming a molecular
  clock, which can be used to obtain relative edge lengths in calendar time.
\end{abstract}

\begin{document}

\maketitle

\section{Introduction}

Phylogenetic networks are models of evolutionary history that extend
phylogenetic trees, by allowing reticulations where one population inherits
genetic material from multiple parent populations.
Reticulations capture the genomic mix that results from
biological processes such gene flow, admixture, hybridization, or
horizontal gene transfer.
Estimation of phylogenetic networks is hard.  Full likelihood and Bayesian
methods, which use genetic sequences as input, do not scale beyond a handful of
taxa.  Methods that use various summaries of the genetic data are faster, but
often do not produce edge lengths in terms of mutations or calendar time.  For
example, \texttt{SNaQ} \citep{SolisLemus2016}, and the
pseudolikelihood method in \texttt{PhyloNet} \citep{yu15} output lengths in coalescent
units, and for a subset of edges only.
Methods using allele frequencies, 
such as \texttt{TreeMix} \citep{pickrell12} and \texttt{ADMIXTOOLS} \citep{patterson12},
provide lengths in ``drift'' units.

It is often desirable to estimate edge lengths in substitutions per
site or calendar time.
The task of estimating edge lengths in calendar time
is traditionally called ``calibration'', and has a long
history of method development in the context when the phylogeny is a species tree.
\citet{r8s} developed the pioneer method \texttt{r8s}
to calibrate trees in absolute time units, allowing for variation of the
molecular substitution rate across lineages to relax the
so-called ``molecular clock'' assumption.
\texttt{BEAST} \citep{beast} has state-of-the-art and
widely-used 
methods for calibrating trees.
Calibrated species trees enable a wide array of downstream analyses,
such as the estimation of 
speciation and extinction rates,
the identification of periods with rapid speciation or extinction across
multiple groups, and the association of diversification rates with
species traits, or geographic and environmental variables.

Challenges with calibrating a phylogenetic tree includes rate variation across
lineages --- violating the molecular clock assumption --- and rate variation
across genes.
Rate variation across lineages is influenced by factors like differences in
generation time. Rate variation across genes can be driven by factors such
as differences in selective pressures.
Some Bayesian methods, such as \texttt{BEAST}, have the capability to handle both types of variation, but
they are computationally heavy and often require sampling a small
number of loci for analysis, to reach convergence and feasible running times.
When the tree topology is known, \texttt{ERaBLE} \citep{binet16} provides a
fast way to estimate edge lengths in substitutions per site using genetic
distances as input.
It focuses on accounting for rate variations across genes, and also
allows for rate variation across lineages by not constraining the output
timed tree to be ultrametric.

The challenges of rate variation across
lineages and across genes also apply to networks.
In addition, new challenges arise, such as
(1) the lack of edge length estimates in substitutions per site, as mentioned above,
and (2) the discordance between gene tree topologies due to reticulations.
Several Bayesian methods for network inference
can infer a rooted network with branch lengths in
expected substitutions per site or in relative time,
such as the \texttt{SpeciesNetwork} package in \texttt{BEAST2} \citep{zhang17}
and the \texttt{MCMC\_SEQ} method in \texttt{PhyloNet} \citep{wen17}.
However, the computational burden remains, and
currently there is a lack of fast methods for the task of calibrating
a phylogenetic network with reticulations.

In this work, we adapt the optimization objective used in \texttt{ERaBLE} to
networks, with the goal of obtaining a method that is both computationally
tractable and consistent.
To address challenge (1) we propose to solve two calibration problems.
The first aims to calibrate the network's edge lengths in
substitutions per site, similar to \texttt{ERaBLE},
with the advantage of not assuming any molecular clock,
treating the network topology as given.
The second problem aims to calibrate the network in units
proportional to time, assuming a molecular clock, also on a fixed network topology.
Challenge (2) can be quite extensive as many groups show high levels of
discordance between gene trees, and gene trees can be difficult to estimate
accurately if mutations are too few or recurrent.
To address this challenge, we propose to avoid gene tree estimation and instead
consider using pairwise distances across multiple genes (which may have
different but unknown trees) to calibrate the network.
Throughout, we allow genes to evolve at different rates.

Prior approaches for networks, such as in \citet{bastide18} and
\citet{karimi19}, use
heuristics to normalize the pairwise distances from each gene
to remove rate variation across genes, and then find node
ages that best fit the resulting average pairwise distances.

In this work we model the observed pairwise distance on a gene as the path
distance on the associated displayed tree of the species network with an
additive error term, which is then scaled by a gene-specific rate.  This model,
which we call the \emph{scaled displayed tree model}, is similar to
the universal pacemaker model in \citet{2012Snir-pacemaker}.

Our proposed method is based on the observation that given a network topology,
there is a linear subspace, which we denote by $\Vdti$, such that pairwise
distance vectors are identical across displayed trees
when projected onto $\Vdti$.
Based on this, the method finds gene rates and network parameters
(edge lengths or node ages) that minimizes a
composite objective $\Qdti$ with two components.
The first component fits, on $\Vdti^\perp$, the average scaled
pairwise distances to the expected average distances given the network
parameters.
The second component fits, on $\Vdti$, the scaled pairwise distances
from each gene to the expected average pairwise distances.
Notably, when the network topology is a tree, the first component
disappears and the method reduces to a version of \texttt{ERaBLE}. 
The proposed criterion $\Qdti$ yields a quadratic programming problem
and can be readily solved by standard methods.
We also show that this method is consistent, in the sense that
if gene-specific distances are estimated
consistently (with noise going to zero) 
then the estimated relative branch lengths and gene rates
converge to the true parameters.

\citet{2024Arasti-TCMM} introduce optimization
problems similar to \texttt{ERaBLE} and to our criterion $\Qdti$. In particular,
their Problem 2, called \texttt{TCMM} (Topology-Constrained Metric Matching)
uses a metric gene tree as input and fits edge lengths
on a given fixed species tree topology to minimize the $\ell_2$-norm between
the gene tree distances and the species tree distances.
This ``per-gene'' approach returns multiple calibrations of the same species tree
when multiple gene trees are available.
These calibrations are summarized as a second step.
However, this approach does not allow for non-tree networks.  It
also differs conceptually from our approach, as we do not infer individual
metric gene trees as an intermediate step.
\citet{bevan05}, a precursor to \texttt{ERaBLE} 
focused on estimating relative rates across genes,
contrasts these two approaches.
Their \texttt{DistR} method could use as input either patristic distances,
which are from gene trees estimated from sequences, or pairwise distances, 
which are directly estimated from sequences.  The approach taken in
\citet{2024Arasti-TCMM} is similar to using patristic distances,
requiring an intermediate step of inferring individual metric gene trees,
unlike in our approach. 
When the individual metric gene
trees can be estimated reliably, such an intermediate step could be beneficial
as it in effect denoises the input.  On the other hand, when the gene trees
are difficult to estimate, our approach has the potential to be more robust.

\section{Notations and problem formulation}

\subsection{Basic notations}

We use definitions from \citet{2023XuAne_identifiability} with some
modifications.

A \emph{topological rooted phylogenetic network}, or ``rooted network'' for short,
on taxon set $X$ is a rooted directed acyclic graph where the leaves
(vertices with out-degree $0$) all have in-degree $1$, and are bijectively
labelled by elements in $X$.  A vertex, also referred to as a node,
is a \emph{tree node} if its in-degree is $0$ or~$1$, and a
\emph{hybrid node} otherwise.  For simplicity, we assume here that
$X = [n]$ where $n$ is the number of taxa in the network.

A rooted network may be endowed with a \emph{metric}, where each edge $e$ is
assigned a length $\ell(e) \geq 0$ and an inheritance parameter
$\gamma(e) \in (0, 1]$, with the restrictions that $\ell(e) > 0$ and
$\gamma(e) = 1$ for tree edges, and $\gamma(e) \in (0, 1)$ for hybrid edges.
Moreover, we require that $\gamma(e)$ sums to $1$ over edges with a common child node.

A node $u$ is \emph{above} node $v$ if there is a direct path from $u$ to $v$.
The \emph{least stable ancestor} (LSA) of a taxon set $Y \subseteq X$ is the lowest
node which all directed paths from the root to nodes labelled by $Y$ must pass.

The \emph{semidirected network} induced from a rooted network is the
semidirected graph obtained by removing all edges and nodes above the LSA of
$X$, suppressing the LSA if is of degree $2$, and undirecting all tree edges
while keeping the direction of hybrid edges.  If there is a metric structure of
the rooted network, it is retained in the induced semidirected network.  When we
refer to a semidirected network, we assume it is induced from some rooted
network.

A rooted or semidirected network $N$ \emph{displays} a tree $T$ (and $T$ is
called a \emph{displayed tree}) by the following process: For every hybrid node,
one parent hybrid edge is kept and the rest deleted, which results in a spanning
tree $\tilde{T}$ of $N$.
Then, we keep only the edges in $\tilde{T}$ the are on paths
between two labelled leaves of $N$
(that is, we iteratively prune any unlabelled leaf), yielding $T$.
Note that the pruning makes the definition of displayed trees different
from that in \citet{2023XuAne_identifiability}, but similar
to the definition in \citet{ane24_anomal_networ_under_multis_coales}.

A network $N$ is of \emph{level 1} if any two distinct cycles
in the undirected graph induced by $N$ do not share any node.

If $N$ has a metric, then $\tilde{T}$ and $T$ inherit the
edge lengths from $N$, and we assign a probability distribution
over displayed trees. Namely, the probability of retaining
a parent edge $e$ in $\tilde{T}$ for a
fixed hybrid node is $\gamma(e)$, with the choice of parent edge independent
across hybrid nodes.
We denote the probability of displaying a tree $T$ by
$\gamma(T)$, and we have $\gamma(T) = \prod_{e \in T} \gamma(e)$.
We denote the set of displayed trees of a semidirected network by
$\mathcal{T}(N)$, or simply $\mathcal{T}$ if no confusion is likely.

\subsection{Notations for calibration}

For reasons given in \Cref{sec:identifiability}, there may be a subset of edges
$E_0(N)$ in $N$ whose lengths are assumed to be $0$.
Then, the subset $E_c(N)$ is defined as the set of edges of $N$ \emph{not}
in $E_0(N)$, also called the \emph{calibrating set}, or set of
\emph{calibrating edges}.
We may omit $N$ in the notation if confusion is unlikely.
In $N$, we enumerate the calibrating edges in $E_c$ as $e_1, \ldots, e_m$, and denote the
\emph{edge length vector} as $b = (\ell(e_1), \ldots, \ell(e_m))$.

For vectors $u, v$ of matching size,
we write $v \succeq u$ if $v_i \geq u_i$
for all $i$, and $v \succ u$ if $v_i > u_i$ for all $i$.
When estimating edge lengths $b$, we may impose the constraint
$b \succeq 0$.

For a metric semidirected network $N$, a calibrating edge set $E_c$,
and for a displayed tree $T$ of $N$,
we write $d_{ij}(T, \hb)$ for the path distance between taxa $i, j$ on $T$
when the lengths of edges in $E_c$ are set to vector $\hb$
(and other edge lengths set to $0$).
For simplicity, we write $d_{ij}(T)$ for
the true path distance $d_{ij}(T, b)$ on $T$.

For a set of pairwise distances $d_{ij}$ over $i, j \in [n]$,
we denote $d$ the vectorization of the distance matrix
obtained by stacking up its columns:
$d = (d_{12}, d_{13}, \ldots, d_{23}, \ldots, d_{(n-1)n}) \in \RR^{\binom{n}{2}}$.
For $T \in \mathcal{T}$ we then have
\[
  d(T, \hb) = A(T)\hb,
\]
where $A(T)$ is the \emph{inheritance matrix} of $T$
with respect to $E_c$ (see Additional file 1 of \citet{binet16}, where
it is called the ``topological matrix'').
It is a binary matrix with
$A_{pq}(T) = 1$ if edge $e_q \in E_c$ is on the path between
the taxon pair corresponding to the $p$-th entry of $d(T, \hb)$,
and $A_{pq} = 0$ otherwise.
We also define the \emph{average distance vector} as
$\bar{d} = \sum_{T \in \mathcal{T}} \gamma(T) d(T)$, and the \emph{average distance
vector given edge lengths} $\hb$ as
\[
  \bar{d}(\hb) = \sum_{T \in \mathcal{T}} \gamma(T) d(T, \hb)\;.
\]
Then $\bar{d}(\hb) = \bar{A}\hb$ and $\bar{d} = \bar{A} b$,
where $\bar{A}$ is the \emph{average inheritance matrix} of $N$:
\[
  \bar{A} = \bar{A}(N) =  \sum_{T \in \mathcal{T}} \gamma(T) A(T)\;.
\]

If the network is an unrooted tree $T$,
then $\bar{A}(T)=A(T)$ contains topological information only.
If $N$ contains hybrid edges,
then $\bar{A}(N)$ encodes information on both the topology and
inheritance probabilities of $N$, but not edge lengths.

\subsection{The unconstrained network calibration problem}

We consider a fixed semidirected network topology $N$,
assumed to be known without error, with known hybridization parameters.
Traditionally, the \emph{calibration problem} consists in estimating
the length $\ell(e)$ of all edges $e$ in $N$, such that $\ell(e)$ is in
unit of calendar time (e.g. millions of years).
We extend this problem to address challenges that arise when $N$
is not a tree. Many network estimation methods output phylogenies that lack
lengths for some or all edges, or have edge lengths in coalescent or drift units,
that correlate poorly with substitutions per site or calendar time
(due to variation in population size or generation time, for example).
In this section we consider the \emph{unconstrained}
calibration problem, where we seek to estimate
lengths in substitutions per site, for edges in the species network.
This unconstrained problem makes no molecular clock assumption.
Later in \Cref{sec:constrained} we briefly consider the
\emph{constrained} calibration problem, where we seek to estimate
edge lengths proportional to time, or equivalently, relative node ages,
and for which we will assume a molecular clock.

In this work, we consider solving calibration problems using data
from pairwise distances in substitutions per site, across multiple genes.
To relate the species phylogeny $N$ with the
genetic pairwise distance $\delta_{ij}^{(k)}$ between taxa $i$ and $j$
observed at locus $k$, we assume
the following \emph{scaled displayed tree model}, which accounts for
variation in substitution rate across genes
and for pairwise distance estimation error.

\begin{model}[Scaled displayed tree model]\label{mod:sdtm}
  Let $N$ be a metric semidirected network on $n$ taxa,
  with set of displayed trees
  $\mathcal{T} = \{T_l: l = 1, \ldots, L\}$ and
  probability $\gamma(T)$ for any displayed tree $T$.
  Then the genetic distance $\delta_{ij}^{(k)}$
  between taxa $i, j \in X$ in gene $G_k$ with $k = 1, \ldots, K$ is given by
  \begin{equation}\label{eq:scaleddisplayetreemodel}
  \delta_{ij}^{(k)} = r_k(d_{ij}(T^{(k)}) + \epsilon_{ijk})
  \end{equation}
  where
  \begin{itemize}
  \item $r_k > 0$ is the rate of gene $G_k$,
    sampled from some distribution of rates across genes;
  \item $T^{(k)}$ is the gene tree associated with gene $G_k$, and
    is a randomly sampled displayed tree,
    with probability $\gamma(T)$ of sampling $T \in \mathcal{T}$;
  \item $\epsilon_{ijk}$ is residual variation, 
    sampled from some distribution with mean $0$.
  \end{itemize}
  The rates $r_k$, gene trees $T^{(k)}$ and
  residuals $\epsilon_{ijk}$ ($k = 1,\ldots,K$) are unobserved,
  and assumed to be sampled independently.
\end{model}

\begin{figure}
  \centering \includegraphics[scale=1]{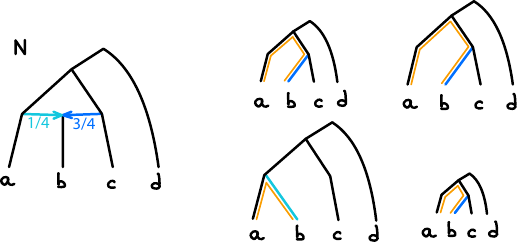}
  \caption{Illustration of the scaled displayed tree model.
  Left: network $N$ on 4 taxa, with hybrid edges annotated by their
  $\gamma$ inheritance values.
  Right: 4 gene trees, randomly sampled from the 2 trees displayed in $N$,
  scaled by their individual rates.
  The observed genetic distances, e.g.\ between $a$ and $b$ (shown in orange),
  are modeled as the path distances on these gene trees with additive error.
  In this example, a short genetic distance between $a$ and $b$ could be the
  result of $a$ and $b$ being sister for that gene due to introgression
  (bottom left gene tree), or a slow rate for the gene (bottom right),
  or residual variation (not illustrated).
  }
  \label{fig:model}
\end{figure}

\Cref{fig:model} illustrates this model: on the left is the network $N$, and on
the right we have a sample of 4 gene trees, scaled by their evolutionary rates.
After adding residual variation, these trees yield the observed genetic distances.
This model is similar to the universal pacemaker model proposed in
\citet{2012Snir-pacemaker}, but differs in that the error term is additive
instead of multiplicative.

Note that we can multiply the edge lengths in $N$ by a constant $c > 0$ and the
rates $r_k$ by $c^{-1}$ without changing the model.
External information about the rates $r_k$ or distances on $N$
may help eliminate this scaling non-identifiability.
For example, knowledge of the genomewide substitution rate per site per year could
provide an informative prior distribution on the rates $r_k$.
Alternatively, knowledge about node ages can constrain the scale of edges lengths
in $N$, such as from fossil data or
from serial sampling with tips of variable and known ages.
But in the absence of such external information and
from distance data,
edge lengths are only identifiable up to a multiplicative factor.
Given this, we interpret the edge length $\ell(e)$ as the
\emph{baseline} substitution rate per site that occurred along $e$ and $r_k$ the
\emph{relative} rate of gene $G_k$, relative to this baseline.
If we also obtain estimates of $r_k$ and
if $N_k$ is the number of sites in gene $k$, then
$$\tilde{\ell}(e) = \frac{\sum_k r_kN_k}{\sum_k N_k} \ell(e)$$
is the genomewide substitution rate per
site that occurred along edge $e$, and is not a relative quantity.

\medskip

When a subset of edges $E_0$ have their lengths assumed to be $0$,
the edge length vector $b$ lists the lengths of the remaining edges
in the calibrating set $E_c$.
The calibration problem then seeks a candidate edge length vector $\hb$
under which pairwise distances predicted by the model
best fit some set of observed distances $\delta_{ij}^{(k)}$.

The exact notion of ``best fit'' would depend on
the distributional assumptions on the
rate and error terms in the model.
In this work, we focus on finding fast procedures that are consistent.

In the following sections we also write
$\alpha = (r_1^{-1}, \ldots, r_K^{-1})$ for the vector of inverse rates,
and $\ha$ for the vector of candidate inverse rates.

Since $b \succeq 0$ and $\alpha \succ 0$, it is natural to add the
corresponding constraints $\hb \succeq 0$ and $\ha \succ 0$ to optimization
objectives for the calibration problem.  This is desirable because it ensures
the resulting lengths and scales are interpretable.  However, in the following
section we state our results without these constraints for simplicity, as a
consistent procedure without such constraints
remains consistent with the constraints in place.  In addition, the proposed
criteria \eqref{eq:qdti} and \eqref{eq:qdti-cons} are quadratic programming
problems \citep{boyd2004} with and without these constraints, and are readily
solved by standard methods.

\section{Tree calibration: \texttt{ERaBLE} is consistent}

Assuming that the phylogeny is a tree $T$, one approach to solve the
calibration problem is \texttt{ERaBLE} \citep{binet16}.
The scaled displayed tree model simplifies to
$\delta^{(k)}_{ij} = r_k (d_{ij}(T) + \epsilon_{ijk})$
where $T^{(k)}=T$ is shared across genes.
This is an explicit form
of the ``proportional model'' in \citet{binet16}.
On a tree, all edges lengths can be estimated from
pairwise distances, that is, the calibrating set $E_c$
contains all edges of $T$.

\texttt{ERaBLE} fits the branch lengths $\hat{b}$ in $T$, along with the inverse
rates $\hat{\alpha}$, by minimizing the following criterion:
\begin{equation}
  \label{eq:erable}
  Q(\hat{\alpha}, \hat{b}) = \sum_k \sum_{i<j}
  w^{(k)}_{ij}(\hat{\alpha}_k \delta^{(k)}_{ij} - d_{ij}(T, \hb))^2
\end{equation}
subject to the constraint $\sum_k z_k \hat{\alpha_k} = \sum_k z_k$.  Here
$w^{(k)}_{ij} > 0$ are weights and
$z_k\neq 0$ are constants
that do not depend on $\ha$ or $\hb$.
\citet{binet16} used $w_{ij}^{(k)} = N_k$ and
$z_k = N_k \sum_{i < j} \delta_{ij}^{(k)}$, where $N_k$ is the sequence length
of gene $k$.  But other choices are possible, as noted in \citet{binet16}.

We show here that \texttt{ERaBLE} is consistent, in the sense that when the error in the
input distances tends to $0$, the solutions to \eqref{eq:erable} converge to the
true parameter, up to a multiplicative factor.  (We are unaware of any prior
proof.)
The assumption of input distance error going to $0$
can be interpreted as requiring the sequence length of each gene to tend to
infinity (given a consistent estimator of the gene-specific distance from
sequences), while the number of genes stays fixed.

\begin{theorem}
  \label{thm:erable}
  Let observed distances come from the scaled displayed tree
  model, under a network $N=T$ that is a tree.
  Given a fixed number of genes $K$, assume that
  for each taxon pair $i,j$ and for each gene $k$,
  $\epsilon_{ijk} \to 0$ almost surely.
  Then any solution $(\as, \bs)$ of~\eqref{eq:erable}
  converges to $(c\alpha, cb)$ almost surely, where $\alpha = (r_1^{-1}, \ldots, r_K^{-1})$,
  and $c = \dfrac{\sum_k z_k}{\sum_k\frac{z_k}{r_k}}$ is the weighted
  harmonic mean of the unknown true gene rates, with weights given by $z$.
\end{theorem}

\begin{proof}
  We first show that when $\epsilon_{ijk} = 0$, the unique solution
  is $(c\alpha, cb)$.
  Note that 
  Since $Q \geq 0$, any choice of $\hat{\alpha}, \hb$ that makes $Q = 0$ is
  a minimizer of~\eqref{eq:erable}.
  $Q = 0$ if and only if
  $\hat{\alpha}_k \delta^{(k)}_{ij} = d_{ij}(T, \hb)$ for all $k, i, j$,
  that is,
  $\hat{\alpha}_kr_k = \frac{d_{ij}(T, \hb)}{d_{ij}(T)}$
  given the assumption that $\delta_{ij}^{(k)} = r_k d_{ij}$
  with no estimation error.
  This is equivalent to there being a constant $c$ such that
  $\hat{\alpha}_k r_k = c$ for all $k$, and
  $\frac{d_{ij}(T, \hb)}{d_{ij}(T)} = c$ for all $i,j$.
  Then we must have that $\hat{\alpha}_k = c r_k^{-1}$ with
  $c = \frac{\sum_k z_k}{\sum_k\frac{z_k}{r_k}}$ to satisfy
  the linear constraint
  $\sum_k z_k\hat{\alpha}_k = \sum_k z_k$.
  Then choosing $\hb = cb$ gives
  $d_{ij}(T, \hb) = c d_{ij}(T)$ for all $i,j$, proving the existence
  of $\hat{\alpha}, \hb$ such that $Q(\hat{\alpha}, \hb)=0$.
  Since the path distances on a phylogenetic tree
  determine the branch lengths uniquely,
  $d_{ij}(T, \hb) = c d_{ij}(T)$ implies that $\hb = cb$,
  thereby proving the uniqueness of $(\alpha^*,b^*)$ minimizing $Q$.

  Now we consider the scenario when $\epsilon_{ijk} \to 0$ almost surely.
  To emphasize the dependence of $Q$ on the data
  via the vector of all residual terms $\epsilon={(\epsilon_{ijk})}_{i,j,k}$,
  we write $Q(\ha, \hb)$ here as
  $Q(\ha, \hb \,; \epsilon)$.
  Since $Q(\ha, \hb\,; \epsilon)$ is convex quadratic in $(\ha, \hb)$
  for any $\epsilon$, there always exist some minimizer, although not
  necessarily a unique one.  We shall show that any sequence of minimizers (even
  when some of the minimizer is not unique) converges to $(c\alpha, cb)$.  By
  Theorem~2 in \citet{kall86}, it suffices to show that
  $Q(\ha, \hb \,; \epsilon)$ converges uniformly to $Q(\ha, \hb \,; 0)$
  on any compact subset of the domain as $\epsilon \to 0$.  This can be easily shown using elementary
  bounds, and we provide details in \Cref{sec:techproofs}.
\end{proof}

\section{Unconstrained network calibration}

\subsection{Identifiability}
\label{sec:identifiability}

The unconstrained problem is formulated with semidirected networks,
in which tree edges are undirected and the root position is
not specified, because the root position is not
identifiable from pairwise distances alone.
As noted in
\citet{2023XuAne_identifiability}, the paths in displayed trees of a rooted
network, or equivalently up-down paths in the network, are determined by the
semidirected structure of the network.  Since the pairwise distances derive from
such paths, one cannot identify beyond the semidirected structure with such
data.
This however does not preclude one from calibrating a rooted topology
as done later in \Cref{sec:constrained}, in which the root position is
obtained from external information (such as a priori knowledge of
outgroup species) or from other types of data.

On a given semidirected network, it may not be possible to identify all the edge
lengths using pairwise distances.
For example, if a hybrid node has a unique child edge, one may add
$\epsilon$ to the length of the child edge and $-\epsilon$ to the lengths of all
its parent hybrid edges, while leaving unchanged the gene-specific pairwise
distances between taxa on each displayed trees.  This is referred to as
``unzipping'' in \citet{Pardi_2015}, and is also mentioned in \citet{willson13}.
One may get around this identifiability issue by requiring the hybrid's child
edge to have length $0$.  In that case, we would define $E_0$ as the set of
single-child edges of hybrid nodes, and the calibration set $E_c$ as containing
all other edges.

The lack of edge length identifiability is more extensive
when the data consist of average distances.  At a given hybrid node,
one may set to $0$ the length of all its parent hybrid edges and lengthen
its child edge, while leaving all average pairwise distances unchanged.
This is referred to as
``zipping-up'' a reticulation \citep{2023XuAne_identifiability}.
Later, we will consider a criterion that calibrates edge lengths by
fitting average distances. For this criterion, we will require all hybrid edges
to have length $0$, that is, we will restrict the calibrating set $E_c$ to
consist of tree edges only.

\subsection{From naive to consistent extensions}

In this section we consider a few naive but natural approaches that have various
drawbacks.  They will motivate the approach proposed in the next section.

The network calibration problem is challenging due to the variability of
the displayed trees $T^{(k)}$.
If we knew the tree topology of $T^{(k)}$, which the distance data $\delta^{(k)}$
is coming from, then we could optimize an aggregate of \texttt{ERaBLE} objectives, e.g.\
minimize
$\sum_k \| \ha_k \delta^{(k)} - d(T^{(k)}, \hb) \|^2$
over $(\ha, \hb)$
with some linear constraint on $\ha$
(and choosing $w_{ij}^{(k)}=1$ for simplicity).

Without knowledge of the gene tree topologies,
one could proceed with finding, for each $k$, the displayed tree that
best fits each distance vector from gene $k$,
e.g.\ using the following criterion instead:
$$\sum_{k=1}^K \min_{T \in \mathcal{T}}\| \ha_k \delta^{(k)} - d(T, \hb) \|^2.$$
However, each evaluation of this criterion involves searching over all displayed
trees for each gene, making it expensive to optimize,
as the number of displayed tree in $N$ grows exponentially
with the number of reticulations.

A naive approach that maintains the computational tractability
of the objective function is to replace the minimum over displayed trees
by an expectation, that is
\[
  Q_1(\ha, \hb) = \sum_{k=1}^K \sum_{T \in \mathcal{T}}
   \gamma(T) \; \| \ha_k \delta^{(k)} - d(T, \hb) \|^2.
\]
But this turns out to give a biased estimate.
Consider the 5-sunlet in \Cref{fig:q1-example}, with
$\gamma(e_1)=\gamma(e_2)=0.5$, and with all tree edges
having length $1$ and both hybrid edges having lengths $0$.
Let $E_c$ be the set of tree edges.
Consider solving the unconstrained calibration problem using
2 genes, with gene 1 evolving under the displayed tree $T^{(1)}$ that keeps
hybrid edge $e_1$, and gene 2 evolving under the displayed tree $T^{(2)}$
with edge $e_2$. This simple choice implies the best case scenario when the
proportion of genes with tree $T^{(i)}$ matches $\gamma(e_i)=0.5$.
Suppose we observe the distance vector from each gene tree with no error,
and that genes have inverse rates $r_1^{-1} = 0.5$ and $r_2^{-1} = 1.5$.
Their average is 1, that is: $z^\top r^{-1} = 1$ where we choose
$z = (\frac{1}{2}, \frac{1}{2})$.
Optimizing $Q_1$ with the constraint $z^\top \ha = 1$
yields the estimates $b^*$ shown in \Cref{fig:q1-example}:
the pendant edges are severely overestimated and the
cycle edges are severely underestimated.

\begin{figure}
  \centering \includegraphics[scale=2]{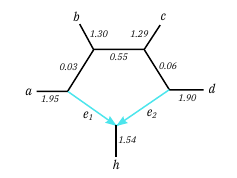}
  \caption{Network with a 5-sunlet topology. Its tree edges are
    annotated with lengths estimated by minimizing $Q_1$,
  when the true edge length are $1$ for tree edges and $0$ for hybrid edges,
  $\gamma(e_1)=\gamma(e_2)=0.5$, and given distance data from
  each displayed tree, with inverse rates $0.5$ and $1.5$.
  The calibration edges, in black, are the tree edges.
  The hybrid edges, in blue, are fixed to $0$.
  On this topology, $Q_1$ overestimates pendant edge lengths
  and underestimates cycle edge lengths.}
  \label{fig:q1-example}
\end{figure}

The reason for such bias becomes evident when we rewrite the criterion:
\begin{equation*}
\begin{split}
  Q_1(\ha, \hb)
  &= K \norm{\bar{\delta}_{\ha} - \bar{d}(\hb)}^2+
    \sum_k\norm{\ha_k \delta^{(k)} - \bar{\delta}_{\ha}}^2 +
    K \sum_T \gamma(T) \norm{\bar{d}(\hb) - d(T, \hb)}^2,
\end{split}
\end{equation*}
where $\bar{\delta}_{\ha} = \frac{1}{K} \sum_k \ha_k \delta^{(k)}$.

The first term of $Q_1$ is fitting the average distance on $N$ to the average
distance of the scaled observed distances.
The second term does not depend on branch lengths $\hb$,
and acts like a penalty that favors $\ha$ such that
$\ha_k \delta^{(k)}$ has low variance across genes.
More importantly, the third term favors $\hb$ such that $d(T, \hb)$ has low
variance across displayed trees $T \in \mathcal{T}$.
Intuitively, this last term is inadequate for a correct estimation
of edge lengths, and causes an estimation bias in $\hb$.

In light of this, we consider using the first term only, and fit the average
distances directly.  Specifically, we consider the following criterion:

\begin{equation}\label{eqn:qavg}
  \Qavg(\ha, \hb) = \norm{\frac{1}{K} \sum_k \ha_k \delta^{(k)} - \bar{d}(T, \hb)}^2
\end{equation}
subject to $z^\top \ha = 1$, where $z \succ 0$.
For this criterion, as discussed in \Cref{sec:identifiability}, we restrict
$E_c$ to consist of tree edges only, in order to retain identifiability.

While $\Qavg$ correctly recovers the edge lengths for the example
network topology in
\Cref{fig:q1-example}, it breaks down for other cases.
Consider the 4-sunlet
in \Cref{fig:qavg-example}, with $\gamma(e_1) = \gamma(e_2) = 0.5$ and with
tree edges of length $1$ and hybrid edges of length $0$.  Let $E_c$ be the set
of tree edges.
Consider the unconstrained calibration problem with 2 genes, with gene $G_1$
(resp.\ $G_2$)
evolving under displayed tree $T^{(1)}$ where $e_1$ (resp.\ $T^{(2)}$
and $e_2$) is kept, which matches $\gamma(e_1) = \gamma(e_2) = 0.5$.
Suppose we observe
the distance vector from each gene with no error and with inverse rates
$r_1^{-1} = 0.5$ and $r_2^{-1} = 1.5$, and minimize $\Qavg$
subject to $z^\top r^{-1} = 1$ where
$z = (\frac{1}{2}, \frac{1}{2})$.
In this case $\Qavg$ has multiple minimizers.
The true parameters
$\alpha = (\frac{1}{2}, \frac{3}{2})$ and $b = (1, \ldots, 1)$
do minimize $\Qavg$ with
$\Qavg(\alpha, b) = 0$
and satisfy $z^\top r^{-1} = 1$.
But so does the choice $\ha = (1, 1)$ and $\hb$ corresponding
to the edge lengths shown in \Cref{fig:qavg-example},
also with perfect fit
$\Qavg(\ha, \hb) = 0$.

\begin{figure}
  \centering \includegraphics[scale=2]{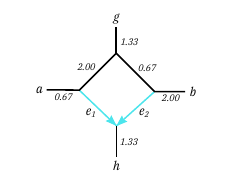}
  \caption{Network with a $4$-sunlet topology. Tree edges are annotated with
    alternative edge lengths $\hb$ that satisfy $\Qavg(\ha, \hb) = 0$ with
    $\ha = (1, 1)$, when the true edge lengths are $1$ for tree edges and $0$
    for hybrid ones, and $\gamma(e_1) = \gamma(e_2) = 0.5$.  The distance data are
    from each displayed tree with inverse rates $0.5$ and $1.5$.
    The calibration edges, in black, are the tree edges.
    The hybrid edges, in blue, are fixed to $0$.
    In this case there is no unique solution to $\Qavg$, and minimizing
    $\Qavg$ does not necessarily recover the true edge lengths.
}

\label{fig:qavg-example}
\end{figure}

Even if we exclude 4-sunlets, it is unclear if $\Qavg$ would be consistent for
simple classes of networks such as the class of level-1 networks
without 4-cycles.  In \Cref{sec:qavg} we characterize when $\Qavg$ can be consistent and
discuss the implications.

\subsection{A composite and consistent criterion}
\label{sec:dti}

We again consider an ideal scenario with 
perfect knowledge of displayed tree frequencies and consistent
estimation of gene-specific distances from infinitely long sequences.
This situation provides a minimal requirement that a
``good'' criterion should pass, excluding $Q_1$ and $\Qavg$ as seen earlier.

Luckily, there is more structure in the network calibration problem that we can
incorporate into $\Qavg$.  For example, if there is a tree path between taxa $i$
and $j$, then we know that in the absence of error $\alpha_k \delta^{(k)}_{ij}$
should be independent of $k$, for some choice of $\alpha$.
In other words, the distances scaled by $\alpha$ should agree
across genes when there is a tree path between two taxa.

We can further generalize this by considering all linear combinations of
scaled pairwise distances that should agree on all displayed trees, which
leads to the following definition.

\begin{definition}
  Let $N$ be a semidirected network and a calibrating set $E_c$.
  A vector $v \in \RR^{\binom{n}{2}}$ is \emph{displayed tree invariant} if
  $v^\top d(T, \cdot)$ is independent of $T \in \mathcal{T}$.  The linear
  subspace of all such vectors is called the \emph{displayed tree invariant
    subspace}, and we denote it by $\Vdti(N)$, or simply $\Vdti$ when no
  confusion is likely.
  We denote the orthogonal projection onto $\Vdti(N)$ by $\Pidti(N)$ or $\Pidti$.
\end{definition}

The projection $\Pidti$ is interesting because
the projection of distances 
from different displayed trees are identical,
given any fixed candidate edge lengths $\hb$ on $N$.
Indeed, let $B$ be a matrix whose columns form an orthonormal basis $\Vdti$.
Then
$B^\top d(T, \cdot)$ is independent of $T$.  Consequently, for any
$T, T' \in \mathcal{T}$ and $\hb$ we have
$\Pidti d(T, \hb) = BB^\top d(T, \hb) = BB^\top d(T', \hb) = \Pidti d(T', \hb)$,
where $\Pidti$ is the orthogonal projection onto $\Vdti$.

\begin{proposition}
  Let $N$ be a semidirected network with calibrating set $E_c$.
  Then
  $\Vdti = \cap_{T \in \mathcal{T}}\ker(A(T)^\top - A(T_0)^\top)$, where $T_0 \in
  \mathcal{T}$ is any fixed displayed tree and
  $A(T)$ is the topological matrix of $T$ with respect to $E_c$.
\end{proposition}

\begin{proof}
  Recall that $d(T, \hb) = A(T) \hb$.  It follows
  that a vector $v$ is displayed tree invariant if and only if
  $v^\top A(T_0) = v^\top A(T)$ for all $T \in \mathcal{T}$, or equivalently
  $v \in \cap_{T \in \mathcal{T}}\ker(A(T)^\top - A(T_0)^\top)$.
\end{proof}

We now formulate a criterion incorporating displayed tree invariant subspace.
\begin{equation}\label{eq:qdti}
  \Qdti(\ha, \hb) = \norm{(I - \Pidti) \left(\frac{1}{K} \sum_k \ha_k
      \delta^{(k)} - \bar{A} \hb\right)}^2 + \lambda \sum_k \norm{\Pidti(\ha_k
    \delta^{(k)} - \bar{A}\hb)}^2\;.
\end{equation}
We will seek $(\ha, \hb)$ that minimize $\Qdti$
subject to $z^\top \ha = 1$, where $z \succ 0$ and $\lambda > 0$
are pre-specified.
In $\Qdti$, the distance data are decomposed into one component
in $\Vdti$, and one component in the orthogonal complement
$\Vdti^\perp$.
The first term in $\Qdti$ fits the average distance on $\Vdti^\perp$,
and the second term fits the gene-specific distances individually on $\Vdti$.

When $N$ is a tree, then $\Pidti = I$
and $\Qdti$ reduces to \texttt{ERaBLE}.
At the other extreme, if $N$ is sufficiently complex
with $\Vdti=\{0\}$ and $\Pidti = 0$, then $\Qdti$ reduces to $\Qavg$.

We have the following consistency for $\Qdti$.

\begin{theorem}\label{thm:qdti}
  Let $N$ be a metric semidirected network and calibrating set $E_c$, with any
  edge not in $E_c$ having length $0$.
  Suppose that $N$ satisfies the following:
  \begin{enumerate}
  \item $\Pidti \bar{d} \neq 0$;
  \item the edge length vector $b$ is identifiable from
    average distances, given the topology of $N$ and
    the constraint that edges not in $E_c$ have length 0.
  \end{enumerate}
  Let observed distances come from the scaled displayed tree
  model, with a fixed number of genes $K$.  Assume that
  for each taxon pair $i,j$ and for each gene $k$,
  $\epsilon_{ijk} \to 0$ almost surely.
  Also, assume that the gene tree frequencies match the displayed tree probabilities,
  that is,
  $\frac{1}{K}\sum_{k=1}^K \indi(T^{(k)} = T) = \gamma(T)$
  for all $T \in \mathcal{T}$.
  Then any minimizer $(\as, \bs)$ of
  $\Qdti(\ha, \hb)$ subject to $z^\top \ha = 1$ converges to $(c\alpha, cb)$
  almost surely,
  where $c = {(z^\top\alpha)}^{-1} > 0$.
\end{theorem}

\begin{proof}
  We first show that when $\epsilon_{ijk} = 0$, the unique minimizer
  is $(c\alpha, cb)$.
  It is easy to see $\Qdti \geq 0$.  Hence to prove this, we only
  need to show that the solution to $\Qdti(\ha, \hb) = 0$ is
  unique and of the form above.

  Suppose $\Qdti(\hat{\alpha}, \hat{b}) = 0$.
  From the second term of $\Qdti$ we have that
  $\Pidti (\hat{\alpha}_k \delta^{(k)} - \bar{A}\hat{b}) = 0$ for
  $k = 1, \ldots, K$.
  Therefore for all $k$, we have
  $\Pidti \bar{A} \hb = \Pidti \ha_k \delta^{(k)} = \ha_k r_k \Pidti
  d(T^{(k)}) = \ha_k r_k \Pidti \bar{d}$.  The last equality comes from
  the fact that $\Pidti d(T)$ is the same for all $T \in \mathcal{T}$.  Now
  since $\Pidti \bar{d} \neq 0$, we have $\ha_k r_k = \ha_{k'}
  r_{k'}$ for any $k \neq k'$.
  Together with the constraint
  $z^\top \hat{\alpha} = 1$ this yields $\hat{\alpha} = c^{-1}\alpha$, 
  where $c = z^\top \alpha$. 
  Furthermore, $\Pidti (\hat{\alpha}_k \delta^{(k)} - \bar{A}\hat{b}) = 0$
  for each $k$
  implies that $\Pidti (\frac{1}{K} \sum_k \hat{\alpha}_k \delta^{(k)} -
  \bar{A}\hat{b}) = 0$.
  Combined with the first term of $\Qdti$ being $0$, this gives 
  $$\bar{A}\hat{b} = \frac{1}{K} \sum_{k=1}^K \hat{\alpha}_k \delta^{(k)} = \frac{1}{cK} \sum_k
  d_k.$$  

  With the frequency assumption, the observed average
  $\frac{1}{K} \sum_{k=1}^K d_k$ is exactly the
  average $\bar{d}$ expected from $N$,
  and we have $\bar{A} \hat{b} = c^{-1} \bar{d}$.
  With the identifiability assumption, this implies $\hat{b} = c^{-1}b$.

  The case when $\epsilon_{ijk} \to 0$ then follows in a similar way as in the
  proof of \Cref{thm:erable}, which we omit here.
\end{proof}

The frequency assumption
$\frac{1}{K}\sum_{k=1}^K \indi(T^{(k)} = T) = \gamma(T)$ in
\Cref{thm:qdti} may appear restrictive.
Under the scaled displayed tree model, gene trees
are considered to be random, independent, from a
multinomial distribution on displayed trees,
with probabilities $\gamma(T)$ considered fixed and known.
The frequency assumption corresponds to
treating the whole set of sampled gene trees as fixed,
interpreting $\gamma(T)$ as the true \emph{realized} frequency
of displayed trees in the sample, and assuming
and that these frequencies have been estimated without error.
For any sample size, it is trivial to estimate the realized frequencies without
error if we can estimate the tree topology of each sampled gene perfectly,
whereas an infinitely large sample size would be
required to estimate the generative multinomial distribution.

The following proposition shows that the first requirement of
$\Pidti \bar{d} \neq 0$ in \Cref{thm:qdti} is very weak.

\begin{proposition}\label{prop:tree-path}
  Let $N$ be a metric semidirected network.  If $N$ contains a tree path
  of strictly positive length between some pair of leaves, then
  $\Pidti \bar{d} \neq 0$.
\end{proposition}

\begin{proof}
  Let the $m^{\text{th}}$ coordinate of $\bar{d}$ correspond to the taxon pair
  $i,j$ with the above-mentioned tree path $p$ between them in $N$,
  and let $s > 0$ be the length of $p$.
  Then $e_m \in \Vdti$, where $e_m$ is the vector
  with $m^{\text{th}}$ component $1$ and the rest $0$.
  This is because there cannot be any up-down path between taxa
  $i$ and $j$ other than $p$ (all nodes along $p$ must be tree
  nodes except possibly for the source of $p$). Therefore $p$ is in all
  displayed trees and $e_m^\top d(T,\hb)$
  is always the sum $s$ of the edge lengths along $p$, irrespective
  of $T \in \mathcal{T}$.
  Consequently, we have
  $\| \Pidti \bar{d} \| \geq \|e_m^\top \bar{d}\| = s > 0$, which implies
  $\Pidti \bar{d} \neq 0$.
\end{proof}

In fact, \Cref{thm:qdti} would continue to hold if $\Pidti$
was defined differently, as the projection $\Pi_V$ onto a nontrivial subspace
$V$ of $\Vdti$ such that $\Pi_V \bar{d} \neq 0$
(condition 1 in \Cref{thm:qdti}). Indeed, the proof of
\Cref{thm:qdti} only used the fact that $\Pidti \bar{d} \neq 0$, and the
property that $\Pidti d(T, \hb)$ does not depend on $T \in \mathcal{T}$.
For example, the proof of \Cref{prop:tree-path} shows that
$V=\langle e_m \rangle$ would be sufficient, where $m$ corresponds to the taxon pair
linked by a tree path in $N$.
This may simplify computations, as computing $\Pidti$ naively involves
iterating over all displayed trees, the number of which may grow exponentially
with the number of reticulations.
The choice of $V$ could
provide a trade-off between computational cost and accuracy.

\medskip
We can use \Cref{thm:qdti} and \Cref{prop:tree-path} to show that $\Qdti$ is
consistent for level-1 networks, as defined in \citet{2023XuAne_identifiability},
with certain conditions.

\begin{theorem}\label{thm:level-1}
  Let $N$ be a level-1 semidirected network without $2$- or $3$-cycles or
  degree-$2$ nodes, and take its set of tree edges $E_c(N)$ as calibrating set.
  Let observed distances come from the scaled displayed tree
  model, with a fixed number of genes $K$.  Assume that
  for each taxon pair $i,j$ and for each gene $k$,
  $\epsilon_{ijk} \to 0$ almost surely.
  Also, assume that the gene tree frequencies match the displayed tree probabilities:
  $\frac{1}{K}\sum_{k=1}^K \indi(T^{(k)} = T) = \gamma(T)$
  for all $T \in \mathcal{T}$.
  Then any minimizer $(\as, \bs)$ of
  $\Qdti(\ha, \hb)$ subject to $z^\top \ha = 1$ converges to $(c\alpha, cb)$
  almost surely,
  where $c = {(z^\top\alpha)}^{-1} > 0$.
\end{theorem}

To prove \Cref{thm:level-1} we will use the following lemma, proved
in \Cref{sec:techproofs}.

\begin{lemma}\label{lem:tree-path-exist}
  Let $N$ be a level-1 semidirected network without $2$-cycles and with at least
  $2$ leaves. Then there exists a tree path between some pair of leaves of $N$.
\end{lemma}

\begin{proof}[Proof of \Cref{thm:level-1}]
  If $N$ has a single leaf, then $E_c(N) = E(N) = \emptyset$ and the statement
  is trivial. Now assume that $N$ has at least 2 leaves.
  By \Cref{lem:tree-path-exist}, $N$ has a tree path between some pair of leaves,
  so by \Cref{prop:tree-path}, condition~1 of \Cref{thm:qdti} is met.
  To apply \Cref{thm:qdti},
  it suffices to show that condition~2 is also met, that is,
  $b \succeq 0$ is identifiable from average distances.  Since
  this is equivalent to $\bar{A}$ being full-rank, it suffices to show that
  $b$ is identifiable when restricted to an open subset.
  As a result, we shall assume $b \succ 0$ and show it is identifiable from average distances.

  Since $E_c$ consists of all the tree edges, the hybrid edges are fixed to have
  length $0$.  In other words, we only consider zipped-up networks
  \citep{2023XuAne_identifiability}.  By Theorem~12 in
  \citet{2023XuAne_identifiability}, from the average distances we can identify
  the mixed network representation of $N$, which contains information about both
  edge lengths and the topology of the network.
  In particular, the mixed network representation yields the length for all edges
  that are not in or incident to a $4$-cycle.  With knowledge of
  the topology of $N$ and its hybridization parameters, the lengths of remaining
  edges can be obtained by undoing the mixed network representation of
  $4$-cycles, by solving for $\mu_a, \mu_b, \mu_g, \mu_h, t_a, t_b$ in
  Equation~(10) in \citet{2023XuAne_identifiability}.  This establishes the
  identifiability of $b$.  The conclusion follows by \Cref{thm:qdti}.
\end{proof}

\section{Constrained network calibration}
\label{sec:constrained}

The \emph{constrained} calibration problem uses as input a
metric rooted network $N^+$ with
known topology and hybridization parameters.
We assume that the true length of an edge $e=(u,v)$ in $N^+$
is derived from the \emph{age} of nodes $u$ and $v$, as
$\ell(e) = \tau(u) - \tau(v)$ where $\tau$ is a \emph{node age function}.
To ensure that $\ell(e)\geq 0$, $\tau$ must satisfy the natural
condition that a child is no older than its parent:
\[\tau(u) \geq \tau(v) \mbox{ for every edge } e=(u,v) \mbox{ in } N^+. \]
In this case $N^+$ is \emph{time-consistent}, in that all paths between two fixed nodes
$u$ and $v$ share the same length, here $\tau(u)-\tau(v)$.
This is consistent with branch lengths being proportional to time.

In addition, for any leaf $v$, $\tau(v)$ should be fixed to the
known age of $v$, the sampling time at which $v$ was observed.
For example, if all leaves are sampled at the present-day and therefore
share the same age, it is typical to anchor time with age 0 at present-day:
fixing $\tau(v)=0$ for all leaves $v$ in $N^+$.

As before, we observe genetic distances
$\delta_{ij}^{(k)}$ between taxa $i, j \in [n]$
from genes $k = 1, \ldots, K$ from the scaled
displayed tree model.
In the constrained calibration problem,
we wish to infer the node ages $\tau(v)$ for all internal (non-leaf) nodes
$v$ of $N^+$.
In effect, in this constrained problem, we assume a molecular clock for
each gene, because $N^+$ is assumed time-consistent and the
gene's expected edge lengths in~\eqref{eq:scaleddisplayetreemodel}
are proportional to edge lengths in $N^+$.


We enumerate the nodes in $V$ as $v_1, \ldots, v_{|V|}$ and denote the
\emph{node age vector} as $t = (\tau(v_1), \ldots, \tau(v_{|V|}))$.  
The requirements on $\tau$ can then be expressed as
$C_1 t \succeq 0$ and $C_2 t = t_0$, where $C_1$ and $C_2$ are matrices with
$C_1 t \succeq 0$ encoding $\tau(u) \geq \tau(v)$ for any edge $(u, v)$ and
$C_2t = t_0$ encoding leaf ages being fixed to their known sampling time $t_0$.

We may add more assumptions to the network, such as assuming hybrid edges having
length $0$, or equivalently that $\tau(u) = \tau(v)$ for a hybrid edge $(u, v)$.
Such assumptions can be reflected by modifying the constraint matrix $C_2$.
Since the assumption that certain edges have length $0$ can be encoded in $C_2$,
we no longer need to specify a calibrating set $E_c$ as in the constrained problem.

In this section we focus on the case when $t_0 = 0$, i.e.\ all the
leaves are sampled at present day.
In this case, similar to the
unconstrained calibration problem, the node ages are only identifiable up to a
multiplicative factor, and are interpreted as relative ages.  This is natural
when we lack fossil data, serially-sampled leaves,
and knowledge about substitution rate per year.

\bigskip

We now turn to adapting the composite criterion $\Qdti$
to the constrained calibration problem.
With a slight abuse of notation, in this section we use $d(T, \htt)$ (and later
similarly $\bar{d}(\htt)$ and $\Qdti(\ha, \htt)$) to denote the vector of
pairwise distances on displayed (rooted) tree $T \in \mathcal{T}$ of $N^+$ if the node
age vector is set to $\htt$.  We then have
\[ d(T, \htt) = B(T)\htt,\]
where $B(T)$ is a matrix with $B(T)_{pq} = 2$ if the node $v_q$ is the most
recent common ancestor, in tree $T$, 
of the taxon pair corresponding to the $p^{\text{th}}$
entry of $d(T, \htt)$, and $B(T)_{pq} = 0$ otherwise.
%
We similarly define
$$\bar{B} = \bar{B}(N) = \sum_{T \in \mathcal{T}} \gamma(T) B(T)$$
and
$$\bar{d}(\htt) = \sum_{T \in \mathcal{T}}d(T, \htt) \;,$$
from which we have $\bar{d}(\htt) = \bar{B}\htt$.

Since distance vectors on displayed trees are linear transformations of
the node age vector, the average distance vector is also linear in node ages,
and the $\Qdti$ criterion can be easily adapted.
By replacing $\bar{A}\hb$ by $\bar{B}\htt$ and adding the
constraints on $\htt$, we obtain the criterion below, for some
fixed $\lambda > 0$.
\begin{equation}
  \label{eq:qdti-cons}
  \Qdti(\ha, \htt) = \norm{(I - \Pidti) \left(\frac{1}{K} \sum_{k=1}^K \ha_k
      \delta^{(k)} - \bar{B} \htt\right)}^2 +
  \lambda \sum_{k=1}^K \norm{\Pidti(\ha_k \delta^{(k)} - \bar{B}\htt)}^2\,.
\end{equation}
We will seek to minimize $\Qdti(\ha, \htt)$
subject to the following linear constraints,
given some pre-specified $z \succ 0$:
$z^\top \ha = 1$, $C_1\htt \succeq 0$, and $C_2 \htt = 0$.

We have the following parallel to \Cref{thm:qdti}.  The proof proceeds in the
same manner and is omitted for brevity.

\begin{theorem}\label{thm:qdti-cons}
  Let $N^+$ be a metric rooted network whose node age vector $t$
  satisfies $C_1 t \succeq 0$ and $C_2 t = 0$.
  Suppose that $N^+$ satisfies the following:
  \begin{enumerate}
  \item $\Pidti \bar{d} \neq 0$;
  \item the node age vector $t$ is identifiable from average distances, given
    the topology of $N^+$ and
    the constraints $C_1 t \succeq 0$ and $C_2 t = 0$.
  \end{enumerate}
  Let observed distances come from the scaled displayed tree
  model, with a fixed number of genes $K$.  Assume that
  for each taxon pair $i,j$ and for each gene $k$,
  $\epsilon_{ijk} \to 0$ almost surely.
  Also, assume that gene tree frequencies match displayed tree probabilities,
  that is,
  $\frac{1}{K}\sum_{k=1}^K \indi(T^{(k)} = T) = \gamma(T)$
  for all $T \in \mathcal{T}$.
  Then any minimizer $(\as, \ts)$ of
  $\Qdti(\ha, \htt)$ subject to $z^\top \ha = 1$ converges to $(c\alpha, ct)$
  almost surely,
  where $c = {(z^\top\alpha)}^{-1} > 0$.
\end{theorem}

\medskip

We also have a parallel to \Cref{thm:level-1}.

\begin{theorem}\label{thm:level-1-cons}
  Let $N^+$ be a rooted level-1 network without $2$- or $3$-cycles, and
  without degree-$2$ nodes except possibly the root,
  and with the root being the LSA of $N^+$.
  Suppose $N^+$ is time-consistent with node age vector $t$.
  Suppose further that each hybrid has the same age as its parents,
  that is, the node age function
  $\tau$ satisfies $\tau(u) = \tau(v)$ for any
  hybrid edge $(u, v)$.
  Let $C_1$ and $C_2$ be constraint matrices
  that encode the requirements
  $C_1 t \succeq 0$ for non-negative edge lengths and
  $C_2 t = 0$ for leaf ages and hybrid ages.

  Let observed distances come from the scaled displayed tree
  model, with a fixed number of genes $K$.  Assume that
  for each taxon pair $i,j$ and for each gene $k$,
  $\epsilon_{ijk} \to 0$ almost surely.
  Assume that gene tree frequencies match displayed tree probabilities,
  that is,
  $\frac{1}{K}\sum_{k=1}^K \indi(T^{(k)} = T) = \gamma(T)$
  for all $T \in \mathcal{T}$.
  Then any minimizer $(\as, \ts)$ of
  $\Qdti(\ha, \htt)$ subject to $z^\top \ha = 1$ converges to $(c\alpha, ct)$
  almost surely,
  where $c = {(z^\top\alpha)}^{-1} > 0$.
\end{theorem}

\begin{proof}
  Using \Cref{thm:qdti-cons} and \Cref{prop:tree-path}, we only need to show
  that the node age vector $t$ is identifiable from average distances.
  Let $r$ be the root of $N^+$, and $E_r$ the set of edges incident to
  $r$.  Consider $N$, the semidirected network induced from $N^+$.
  $N$ and $N^+$ have the same average distances.  Furthermore, with the possible
  exception of edges in $E_r$ when $r$ is of degree-$2$ and gets suppressed in
  $N$, each edge in $N^+$ corresponds to an edge of the same length in $N$.

  By the proof of \Cref{thm:level-1}, edge
  lengths in $N$ are identifiable from average distances.
  If $r$ has degree $3$ or higher in
  $N^+$, this determines all the edge lengths in $N^+$ and in turn the node
  ages.
  If $r$ has degree $2$, let $e_1$ and $e_2$ be its incident edges.
  All other edges in $N^+$ have their lengths determined,
  and $\ell(e_1) + \ell(e_2)$ is also determined.
  Let $x_1$ (resp.~$x_2$) be a leaf descendant of $e_1$ (resp.~$e_2$),
  and $p_1$ (resp.~$p_2$) a path from $r$ to $x_1$ (resp.~$x_2$).
  Then $\ell(e_1)$ and $\ell(e_2)$ are determined by the constraint that
  $p_1$ and $p_2$ have the same length, $\tau(r)$.  This again
  determines all the node ages.
\end{proof}

\section{Discussion}

The main advantages of the scaled displayed tree model
is its accommodation of rate variation across genes, with evolutionary rate
$r_k$ for gene $k$, and its accommodation of rate variation across lineages,
leveraged in the unconstrained calibration problem when the network is
not required to be time-consistent.
The scaled displayed tree model has its limitations.
In particular, it does not allow for
gene-by-linage interactions \citep{rodriguez-trelles01,bedford08}
on evolutionary rates, which occurs when
the relative rates of genes vary from one lineage to another.
In practice, methods such as \texttt{ClockstaRX} \citet{2024Duchene-ClockstaRX} can be
used to determine if 
the scaled displayed tree model is violated,
and if so, to partition the data into groups of genes that follow
a single ``pacemaker''. Network calibration assuming the scaled displayed tree
model could then proceed on each group of genes separately, and the different
calibrations could be summarized as a second step, similarly to the 2-step
procedure proposed by \citet{2024Arasti-TCMM} to calibrate a species tree.

\medskip

The universal pacemaker model was formulated by
\citet{2012Snir-pacemaker} in terms of edge lengths in
gene trees with multiplicative residual variation whose log-normal
distribution. In terms of pairwise distances, an analogous model is:
\begin{equation}\label{eq:pacemakerdist}
 \delta_{ij}\sk = r_k d_{ij}(T\sk) \, \eta_{ijk}
\end{equation}
where $\eta_{ijk}$ are independent log-normal random variables with mean $1$,
and where gene $k$ evolves under some displayed tree $T\sk$,
extending the model by \citet{2012Snir-pacemaker} to allow for
gene tree discordance.
We can rewrite~\eqref{eq:pacemakerdist} as
\[
  \log(\delta_{ij}\sk) = \log(r_k) + \log(d_{ij}(T\sk)) + \epsilon_{ijk},
\]
where $\epsilon_{ijk}$ are independent, normally distributed and with mean $0$.
A corresponding least-squares criterion could be formulated, but it is
unclear how averaging across genes $k$ should be handled
when using log-transformed distances.
Nonetheless, our results show that $\Qdti$ is consistent under the universal
pacemaker model \eqref{eq:pacemakerdist},
because it is equivalent to the scale displayed tree model
\eqref{eq:scaleddisplayetreemodel} with
$\epsilon_{ijk} = d_{ij}(T\sk) (\eta_{ijk} - 1)$.
As long as $\eta_{ijk} \to 1$ almost surely,
the estimate from minimizing $\Qdti$ is consistent.

\medskip 

When using the criterion 
proposed here, one may wish to account for
variable estimation error across genes or across taxon pairs.
Such heterogeneous variance of $\epsilon_{ijk}$ in the
scaled displayed tree model could be accommodated using
different weights into the criterion.
For example, estimated distances should be more precise on
long genes than short genes. To account for this variable precision across
genes, we can modify $\Qdti$ to obtain a weighted criterion $\Qdti^W$ similar to \citet{binet16},
by weighing the squared error on gene $k$ by $N_k$,
the alignment length of gene $k$, as follows:
\[
    \Qdti^W(\ha, \hb) = \norm{(I - \Pidti) \left(\frac{1}{K} \sum_k \ha_k
      \delta^{(k)} - \bar{A} \hb\right)}^2 + \lambda \sum_k N_k \norm{\Pidti(\ha_k
    \delta^{(k)} - \bar{A}\hb)}^2\;.
\]
In $\Qdti^W$, the first term is the same as in $\Qdti$, and the second term has
weights $N_k$ to account for gene lengths.

Accounting for different variances across taxon pairs is more challenging.
Theoretically, if $\var(\epsilon_{ijk}) = \sigma_{ij}^2$ does not depend on $k$,
then we may modify the $\ell_2$-norm used in $\Qdti$ to the $\|\cdot\|_W$ norm,
defined as $\|x\|_W^2 = x^\top W x$ for $x \in \RR^{\binom{n}{2}}$, where $W$ is
the diagonal matrix with $\sigma_{ij}^{-2}$ on the diagonal.  The projection
$\Pidti$ also needs to be adapted, to project onto $\Vdti$ with respect to the
inner product $\langle x, y \rangle_W = x^\top W y$.
However, it is unclear how to estimate $\sigma_{ij}$.
Standard methods can estimate the variance of 
$\delta_{ij}\sk$ from the sequence alignment \citep{bulmer91}, 
but here we need the variance of the scaled genetic distances $\ha_k \delta_{ij}\sk$.
It is also unclear how the criterion can be adapted when
$\var(\epsilon_{ijk})$ also depends on $k$.
\medskip 

In practice, the choice of $\lambda$ in the definition of $\Qdti$
affects the estimated calibrated edge lengths. This choice
corresponds to a trade-off between
fitting the scaled distances from each gene on $\Vdti$ versus fitting the
average scaled distances on $\Vdti^\perp$.
Since the stochastic variance of the average distance 
is smaller than that of the distance from an individual gene
by a factor $1/K$ where $K$ is the number of genes,
$\lambda = K^{-1}$ is a reasonable choice.  An optimal choice of
$\lambda$ would depend on properties beyond consistency, and requires further
investigation.

\section*{Acknowledgements}

We thank Celine Scornavacca for motivating discussions about \texttt{ERaBLE}
and the need to extend it to networks.

\section*{Fundings}

This work was supported in part by the National Science Foundation
(DMS 2023239) and by a H. I. Romnes faculty fellowship
to C.A. provided by the University of Wisconsin-Madison Office of the
Vice Chancellor for Research with funding from the
Wisconsin Alumni Research Foundation.

\section*{Conflict of interest disclosure}

The authors declare that they comply with the PCI rule of having no financial conflicts of interest in relation to the content of the article. 

\appendix
\appendixpage
\renewcommand{\thesection}{\Alph{section}} 

\newcommand{\tb}{\tilde{b}}
\newcommand{\tbe}{\tilde{\beta}}
\newcommand{\ta}{\tilde{\alpha}}
\newcommand{\gt}{\gamma_{\mathcal{T}}}

\section{Fitting average distance directly}
\label{sec:qavg}

In this section we characterize when $\Qavg$ can be consistent,
in the weak sense that its output estimates are correct in the
absence of noise in the input, and discuss implications.
First, when optimizing $\Qavg$, the scales $\alpha_k$, or
equivalently the gene-specific rates $r_k$, cannot be estimated individually.
This is
because when a displayed tree is sampled multiple times, say
$T^{(i)} = T^{(j)} \in \mathcal{T}$, then
$\frac{1}{K} \sum_k \ha_k \EE\delta^{(k)} = \frac{1}{K} \sum_k \ha_k r_k d(T^{(k)},
b)$ yields the same expected average distance estimate
as long as $\ha_i r_i + \ha_j r_j$ remains the same.
Consequently we will focus on when we can recover the edge lengths $b$.

In the following theorem, we characterize when $\Qavg$ returns the correct edge
lengths under optimal conditions.  We establish some notations first.

\bigskip
We enumerate the set of displayed trees $\mathcal{T}$ in some fixed order
$T_1, \ldots, T_L$, then collect their probabilities in vector
$\gt = (\gamma(T_1), \ldots, \gamma(T_L))$, and collect their distance vectors
$d(T), T \in \mathcal{T}$ as columns to form the $\binom{n}{2} \times L$ matrix
$\Ddisp = [d(T_1) \ldots d(T_L)]$.
Then $\bar{d}$ can be written as
$\bar{d} = \bar{A} b = \Ddisp \gt$.

For a matrix $M$, we use $\im^+(M) = \{Mx | x \succeq 0\}$ to denote the cone
contained in the image $\im(M)$ of $M$.  For a cone $C \in \RR^k$, let
$\dim(C)$ be the dimension of the smallest linear subspace that contains $C$.

\begin{theorem}
  \label{thm:qavg}
  Let $N$ be metric semidirected network with calibrating set $E_c$ and edge
  length vector $b \neq 0$.
  Suppose that $\epsilon_{ijk} = 0$, that
  $\bar{A}$ has full column rank, and that all displayed trees are
  sampled, i.e.\ $\{T^{(k)}: k = 1, \ldots, K\} = \mathcal{T}$.

  Let $z \succ 0$ be a fixed vector of length $K$.  Then $\bs = cb$ for some
  $c > 0$ holds for any minimizer $(\as, \bs)$ of $\Qavg(\ha, \hb)$ in
  \eqref{eqn:qavg}, subject to $z^\top \ha = 1$, $\ha \succeq 0$, and
  $\hb \succeq 0$, if and only if
  $\dim\left( \im^+(\Ddisp) \cap \im^+(\bar{A}) \right) = 1$.
\end{theorem}

\begin{proof}
  Let $V = \im^+(\Ddisp) \cap \im^+(\bar{A})$. By the definition of
  semidirected networks, every edge appears in some up-down path
  (or equivalently a path on a displayed tree) between some pair of taxa.
  Consequently, since $b \neq 0$, we have $d(T) = A(T) b \neq 0$ for some $T \in
  \mathcal{T}$, which implies
  $\bar{d} = \Ddisp \gt = \bar{A} b \neq 0 $, and necessarily
  $\dim(V) \geq 1$.

  Without loss of generality, assume $z^\top \alpha = 1$.  Clearly
  $\Qavg(\alpha, b) = 0$ and $(\alpha, b)$ is a minimizer to the optimization
  problem, under the assumption that all $\epsilon_{ijk} = 0$.
  Let $(\as, \bs)$ be a minimizer to the optimization problem. Then
  $\Qavg(\as, \bs) = 0$ and 
  $\frac{1}{K} \sum_k \as_k \delta^{(k)} = \bar{d}(T, \bs) = \bar{A} \bs$.
  Now we have
  \[\frac{1}{K} \sum_k \as_k \delta^{(k)} = 
    \frac{1}{K} \sum_{T \in \mathcal{T}} \left(\sum_{k=1}^K \indi\{T^{(k)} = T\} \as_k r_k\right) d(T, b) =
    \Ddisp \bes,
  \]
  where $\bes \succeq 0$ is the vector defined by
  \[\bes_l = \frac{1}{K} \sum_{k=1}^K \indi\{T^{(k)} = T_l\} \as_k r_k. \]
  Therefore
  $\bar{A} \bs = \Ddisp \bes \in V$.

  First, suppose that $\dim(V) = 1$.
  Then $\bar{A} \bs = c \bar{A} b$ for
  some $c > 0$, and then $\bs = c b$.

  For the converse, suppose that $\dim(V) > 1$.  Then we can find nonzero
  $(\tb, \tbe) \succeq 0$ such that $\bar{A} \tb = \Ddisp \tbe \neq c \bar{A}b$
  for any $c > 0$.  We can also find $\ta \succeq 0$ such that
  $\tbe_l = \frac{1}{K} \sum_k \indi\{T^{(k)} = T_l\} \ta_k r_k$.  Normalizing
  if necessary, we may assume $z^\top \ta = 1$.  We then have a minimizer
  $(\ta, \tb)$ that does not satisfy $\tb = cb$.
\end{proof}

It is important to note that the condition
$\dim\left( \im^+(\Ddisp) \cap \im^+(\bar{A}) \right) = 1$ depends
on the network topology but also on the true
value of $b$, because $\Ddisp$ depends on $b$.  Therefore, in contrast to
\Cref{thm:qdti}, this condition for $\Qavg$ to be consistent cannot be verified
before knowing the solution to the network calibration problem.

The condition is also closely related to identifiability of the metric, i.e.\
$\ell(e)$ and $\gamma(e)$ for edges $e$ of $N$, from average distances,
as detailed below.

If $\dim(\im^+(\Ddisp) \cap \im^+(\bar{A})) > 1$, then we can find $\gt' \succeq
0$ and $b' \succeq 0$ such that $\gt' \not\propto \gt$, $b' \not\propto b$,
${\|\gt'\|}_1 = 1$, and $\bar{A} b' = \Ddisp \gt'$.
Let $N_1$ be the network obtained from $N$ by setting the edge length vector
to $b'$ while keeping $\gamma(e)$ unchanged for all edges $e$.
Suppose we may adjust the inheritance parameters in $N$ so that
$(\gamma(T_1), \ldots, \gamma(T_l))$ becomes $\gt'$ while keeping edge lengths
$b$ as in $N$, and let this network be $N_2$.
Then the average distance vector on $N_1$ is $\bar{A} b'$ and
the average distance vector on $N_2$ is $\Ddisp \gt'$.
Since $\bar{A} b' = \Ddisp \gt'$, $N_1$ and $N_2$
have the same average distances, yet distinct edge parameters.
Therefore, the metric of $N$, given its topology
(the same topology as $N_1$ and $N_2$),
is not identifiable from average distances.

However, the condition on $N_2$ is very strong.
Under the displayed tree probability model
$\gamma(T) = \prod_{e \in T} \gamma(e)$, and so the vector
$\gt = (\gamma(T_1), \ldots, \gamma(T_l))$ is not free to be any
probability vector as we vary 
the values of $\gamma(e)$ for edges $e$ in $N$.
For some network topologies, $\gt$ covers the whole
probability simplex, for example if $N$ has only one hybrid node,
or two hybrid nodes that form a ``hybrid ladder'',
one being a parent of the other.
For these topologies, the argument in the previous paragraph shows
that the identifiability of the metric from average
distances (given knowledge of the network topology) implies that
$\dim(\im^+(\Ddisp) \cap \im^+(\bar{A})) = 1$.
Since such identifiability implies the identifiability of $b$ given
knowledge of both the topology and $\gt$, 
it also implies that $\bar{A}$ is full rank, since $\bar{d} = \bar{A} b$.
Then by \Cref{thm:qavg},
$\Qavg$ recovers true edge lengths in the absence of noise when all displayed
tree are sampled.
However, when $\gt$ does not cover the whole probability simplex,
it is unclear whether the identifiability of the metric given knowledge
of the network topology implies that $\Qavg$ can recover true edge lengths.

\section{Technical proofs}
\label{sec:techproofs}

We first introduce notations that will be used to prove \Cref{lem:tree-path-exist}.
A \emph{blob} $B$ in a rooted or semidirected network $N$ is a
$2$-edge-connected component of $N$, when $N$ is considered as an undirected graph.
The \emph{level} of a blob $B$ is the number of edges in $B$ one needs to remove
in order to obtain a tree.  The \emph{level} of a network is the maximum level
of all its blobs.  This definition of level is closely
related to the standard definition of blobs as biconnected components,
and is used in \citet{2023XuAne_identifiability}.
For binary networks, the two definitions coincide.

\begin{proof}[Proof of \Cref{lem:tree-path-exist}]
  Let $N^+$ be a rooted network obtained from rerooting $N$.  As noted in
  \citet{2023XuAne_identifiability}, the traditional level based on biconnected
  components is lower than or equal to the level defined here,
  therefore all biconnected components of $N^+$ are of level 1,
  as defined in \citet{2010HusonRuppScorn}. Consequently,
  and since $N^+$ has no 2-cycle as assumed \citet{2010HusonRuppScorn},
  it is a tree-child network by Lemma~6.11.11 in \citet{2010HusonRuppScorn}.
  
  We may assume the root $r$ of $N^+$ has at least two tree children: otherwise
  the single tree child $u$ of $r$ must be ancestor to all
  hybrid children of $r$, and we may
  reroot the network at $u$ instead.
  Let $u_1$ and $u_2$ be distinct tree children of $r$.
  Since $N^+$ is tree-child,
  we may take a tree path $p_1$ from $u_1$ to some leaf $x_1$,
  and another tree path from $u_2$ to some other leaf $x_2$.
  Joining them with the root in between forms a tree path
  between $x_1$ and $x_2$, as desired.
\end{proof}

\medskip

\begin{proof}[Details in the proof of \Cref{thm:erable} and \Cref{thm:qdti}]
  For \Cref{thm:erable}, we need to show that $Q(\ha, \hb \,; \epsilon)$
  converges uniformly to $Q(\ha, \hb \,; 0)$ on any compact subset of the
  domain.  It suffices to show uniform convergence when we have
  ${\|\ha\|}_1 \leq R$ and $\|\hb\| \leq R$ for arbitrary $R > 0$.

  To simplify notations, we write $\delta_0^{(k)}$ for the vector
  $r_k d(T^{(k)})$, which is the observed distance vector on gene $k$ when
  $\epsilon = 0$.
  For each gene $k$ we also write $\epsilon^{(k)}$ for the vector
  $(\epsilon_{ijk})_{i < j}$, such that the distance vector from
  gene $k$ is
  $\delta^{(k)} = \delta_0^{(k)} + \epsilon^{(k)}$.

  We can write \eqref{eq:erable} in a compact form:
  \[ Q(\ha, \hb \,; \epsilon) = \sum_k \left(\ha_k \delta\sk - A\hb\right)^\top W\sk
    \left(\ha_k \delta\sk - A\hb\right),
  \]
  where $A$ is the inheritance matrix of $T$, and $W\sk$ is matrix with
  $w_{ij}\sk$ on the diagonal and $0$ off-diagonal.

  Then we have
  \[| Q(\ha, \hb\,; \epsilon) - Q(\ha, \hb\,; 0) | \leq \sum_k \left( 2\ha_k
      \delta_0\sk - 2A\hb + \ha_k\epsilon\sk \right)^\top \left(W\sk \ha_k
      \epsilon\sk\right). \] 

  Let $M = \max_k\| \delta_0\sk \|$ and $W = \max_{i,j,k}w_{ij}\sk$, then for
  each summand, the norm of first term in the inner product is bounded by
  $R(2M + 2 \|A\|_2 + \|\epsilon\|)$, while the second term has norm bounded by
  $W R \|\epsilon\|$.  Uniform convergence follows by application of the
  Cauchy-Schwarz inequality.

  For \Cref{thm:qdti}, we similarly have
  \begin{align*}
    & | \Qdti(\ha, \hb\,; \epsilon) - \Qdti(\ha, \hb\,; 0) | \\
    \leq & \left[(I - \Pidti) \frac{1}{K}\sum_k\left( 2\ha_k
      \delta_0\sk - 2A\hb + \ha_k\epsilon\sk \right)\right]^\top (I - \Pidti)
           \left(\frac{1}{K}\sum_k\ha_k\epsilon\sk\right) \\
         & + \lambda \sum_k \left[\Pidti \left( 2\ha_k
      \delta_0\sk - 2A\hb + \ha_k\epsilon\sk \right) \right]^\top \Pidti
           \ha_k\epsilon\sk\\
    \leq & \left[\frac{1}{K}\sum_k\left( 2\ha_k
      \delta_0\sk - 2A\hb + \ha_k\epsilon\sk \right)\right]^\top
           \left(\frac{1}{K}\sum_k\ha_k\epsilon\sk\right) \\
   & + \lambda \sum_k \left( 2\ha_k
      \delta_0\sk - 2A\hb + \ha_k\epsilon\sk \right)^\top \left(\ha_k\epsilon\sk\right)
     \;.
  \end{align*}
  The norms of the first terms in the $k+1$ inner products are bounded by
  $R(2M + 2 \|\bar{A}\|_2 + \|\epsilon\|)$, while the second terms have norms
  bounded by $R \|\epsilon\|$.
  Again uniform convergence follows by application of the Cauchy-Schwarz inequality.
\end{proof}

\bibliographystyle{plainnat}
\bibliography{ref}

\begin{thebibliography}{25}
\providecommand{\natexlab}[1]{#1}
\providecommand{\url}[1]{\texttt{#1}}
\expandafter\ifx\csname urlstyle\endcsname\relax
  \providecommand{\doi}[1]{doi: #1}\else
  \providecommand{\doi}{doi: \begingroup \urlstyle{rm}\Url}\fi

\bibitem[An{\'e} et~al.(2024)An{\'e}, Fogg, Allman, Ba{\~n}os, and
  Rhodes]{ane24_anomal_networ_under_multis_coales}
C{\'e}cile An{\'e}, John Fogg, Elizabeth~S. Allman, Hector Ba{\~n}os, and
  John~A. Rhodes.
\newblock Anomalous networks under the multispecies coalescent: Theory and
  prevalence.
\newblock \emph{Journal of Mathematical Biology}, 88\penalty0 (3):\penalty0 29,
  2024.
\newblock \doi{10.1007/s00285-024-02050-7}.

\bibitem[Arasti et~al.(2024)Arasti, Tabaghi, Tabatabaee, and
  Mirarab]{2024Arasti-TCMM}
Shayesteh Arasti, Puoya Tabaghi, Yasamin Tabatabaee, and Siavash Mirarab.
\newblock Branch length transforms using optimal tree metric matching.
\newblock \emph{bioRxiv}, 2024.
\newblock \doi{10.1101/2023.11.13.566962}.

\bibitem[Bastide et~al.(2018)Bastide, Sol{\'i}s-Lemus, Kriebel, Sparks, and
  An{\'e}]{bastide18}
Paul Bastide, Claudia Sol{\'i}s-Lemus, Ricardo Kriebel, K~William Sparks, and
  C{\'e}cile An{\'e}.
\newblock Phylogenetic comparative methods on phylogenetic networks with
  reticulations.
\newblock \emph{Systematic Biology}, 67\penalty0 (5):\penalty0 800--820, 2018.
\newblock \doi{10.1093/sysbio/syy033}.

\bibitem[Bedford and Hartl(2008)]{bedford08}
T.~Bedford and D.~L. Hartl.
\newblock Overdispersion of the molecular clock: Temporal variation of
  gene-specific substitution rates in drosophila.
\newblock \emph{Molecular Biology and Evolution}, 25\penalty0 (8):\penalty0
  1631--1638, 2008.
\newblock \doi{10.1093/molbev/msn112}.

\bibitem[Bevan et~al.(2005)Bevan, Lang, and Bryant]{bevan05}
Rachel~B. Bevan, B.~Franz Lang, and David Bryant.
\newblock Calculating the evolutionary rates of different genes: a fast,
  accurate estimator with applications to maximum likelihood phylogenetic
  analysis.
\newblock \emph{Systematic Biology}, 54\penalty0 (6):\penalty0 900--915, 2005.
\newblock \doi{10.1080/10635150500354829}.

\bibitem[Binet et~al.(2016)Binet, Gascuel, Scornavacca, Douzery, and
  Pardi]{binet16}
Manuel Binet, Olivier Gascuel, Celine Scornavacca, Emmanuel J.~P. Douzery, and
  Fabio Pardi.
\newblock Fast and accurate branch lengths estimation for phylogenomic trees.
\newblock \emph{BMC Bioinformatics}, 17\penalty0 (1):\penalty0 23, 2016.
\newblock \doi{10.1186/s12859-015-0821-8}.

\bibitem[Bouckaert et~al.(2019)Bouckaert, Vaughan, Barido-Sottani, Duchêne,
  Fourment, Gavryushkina, Heled, Jones, K{\"u}hnert, Maio, Matschiner, Mendes,
  M{\"u}ller, Ogilvie, du~Plessis, Popinga, Rambaut, Rasmussen, Siveroni,
  Suchard, Wu, Xie, Zhang, Stadler, and Drummond]{beast}
Remco Bouckaert, Timothy~G. Vaughan, Jo{\"e}lle Barido-Sottani, Sebasti{\'a}n
  Duchêne, Mathieu Fourment, Alexandra Gavryushkina, Joseph Heled, Graham
  Jones, Denise K{\"u}hnert, Nicola~De Maio, Michael Matschiner, F{\'a}bio~K.
  Mendes, Nicola~F. M{\"u}ller, Huw~A. Ogilvie, Louis du~Plessis, Alex Popinga,
  Andrew Rambaut, David Rasmussen, Igor Siveroni, Marc~A. Suchard, Chieh-Hsi
  Wu, Dong Xie, Chi Zhang, Tanja Stadler, and Alexei~J. Drummond.
\newblock Beast 2.5: an advanced software platform for bayesian evolutionary
  analysis.
\newblock \emph{PLOS Computational Biology}, 15\penalty0 (4):\penalty0
  e1006650, 2019.
\newblock \doi{10.1371/journal.pcbi.1006650}.

\bibitem[Boyd and Vandenberghe(2004)]{boyd2004}
Stephen Boyd and Lieven Vandenberghe.
\newblock \emph{Convex Optimization}.
\newblock Cambridge University Press, 2004.
\newblock ISBN 9780511804441.
\newblock \doi{10.1017/cbo9780511804441}.

\bibitem[Bulmer(1991)]{bulmer91}
M~Bulmer.
\newblock {Use of the Method of Generalized Least Squares in Reconstructing
  Phylogenies from Sequence Data}.
\newblock \emph{Molecular Biology and Evolution}, 8\penalty0 (6):\penalty0
  868--868, 11 1991.
\newblock ISSN 0737-4038.
\newblock \doi{10.1093/oxfordjournals.molbev.a040696}.

\bibitem[Duch\^{e}ne et~al.(2024)Duch\^{e}ne, Duch\^{e}ne, Stiller, Heller, and
  Ho]{2024Duchene-ClockstaRX}
David~A Duch\^{e}ne, Sebasti\'{a}n Duch\^{e}ne, Josefin Stiller, Rasmus Heller,
  and Simon Y~W Ho.
\newblock {ClockstaRX}: Testing molecular clock hypotheses with genomic data.
\newblock \emph{Genome Biology and Evolution}, 16\penalty0 (4):\penalty0
  evae064, 2024.
\newblock \doi{10.1093/gbe/evae064}.

\bibitem[Huson et~al.(2010)Huson, Rupp, and Scornavacca]{2010HusonRuppScorn}
Daniel~H. Huson, Regula Rupp, and Celine Scornavacca.
\newblock \emph{Phylogenetic Networks}.
\newblock Cambridge University Press, Cambridge, 2010.

\bibitem[Kall(1986)]{kall86}
Peter Kall.
\newblock Approximation to optimization problems: an elementary review.
\newblock \emph{Mathematics of Operations Research}, 11\penalty0 (1):\penalty0
  9--18, 1986.
\newblock \doi{10.1287/moor.11.1.9}.

\bibitem[Karimi et~al.(2019)Karimi, Grover, Gallagher, Wendel, An{\'e}, and
  Baum]{karimi19}
Nisa Karimi, Corrinne~E Grover, Joseph~P Gallagher, Jonathan~F Wendel,
  C{\'e}cile An{\'e}, and David~A Baum.
\newblock Reticulate evolution helps explain apparent homoplasy in floral
  biology and pollination in baobabs (adansonia; bombacoideae; malvaceae).
\newblock \emph{Systematic Biology}, 69\penalty0 (3):\penalty0 462--478, 2019.
\newblock \doi{10.1093/sysbio/syz073}.

\bibitem[Pardi and Scornavacca(2015)]{Pardi_2015}
Fabio Pardi and Celine Scornavacca.
\newblock Reconstructible phylogenetic networks: Do not distinguish the
  indistinguishable.
\newblock \emph{PLOS Computational Biology}, 11\penalty0 (4):\penalty0
  e1004135, 2015.
\newblock ISSN 1553-7358.
\newblock \doi{10.1371/journal.pcbi.1004135}.

\bibitem[Patterson et~al.(2012)Patterson, Moorjani, Luo, Mallick, Rohland,
  Zhan, Genschoreck, Webster, and Reich]{patterson12}
Nick Patterson, Priya Moorjani, Yontao Luo, Swapan Mallick, Nadin Rohland,
  Yiping Zhan, Teri Genschoreck, Teresa Webster, and David Reich.
\newblock Ancient admixture in human history.
\newblock \emph{Genetics}, 192\penalty0 (3):\penalty0 1065--1093, 2012.
\newblock \doi{10.1534/genetics.112.145037}.

\bibitem[Pickrell and Pritchard(2012)]{pickrell12}
Joseph~K. Pickrell and Jonathan~K. Pritchard.
\newblock Inference of population splits and mixtures from genome-wide allele
  frequency data.
\newblock \emph{PLoS Genetics}, 8\penalty0 (11):\penalty0 e1002967, 2012.
\newblock \doi{10.1371/journal.pgen.1002967}.

\bibitem[Rodr{\'i}guez-Trelles et~al.(2001)Rodr{\'i}guez-Trelles, Tarr{\'i}o,
  and Ayala]{rodriguez-trelles01}
Francisco Rodr{\'i}guez-Trelles, Rosa Tarr{\'i}o, and Francisco~J. Ayala.
\newblock Erratic overdispersion of three molecular clocks: Gpdh, sod, and xdh.
\newblock \emph{Proceedings of the National Academy of Sciences}, 98\penalty0
  (20):\penalty0 11405--11410, 2001.
\newblock \doi{10.1073/pnas.201392198}.

\bibitem[Sanderson(2003)]{r8s}
Michael~J. Sanderson.
\newblock R8s: Inferring absolute rates of molecular evolution and divergence
  times in the absence of a molecular clock.
\newblock \emph{Bioinformatics}, 19\penalty0 (2):\penalty0 301--302, 2003.
\newblock \doi{10.1093/bioinformatics/19.2.301}.

\bibitem[Snir et~al.(2012)Snir, Wolf, and Koonin]{2012Snir-pacemaker}
Sagi Snir, Yuri~I. Wolf, and Eugene~V. Koonin.
\newblock Universal pacemaker of genome evolution.
\newblock \emph{PLOS Computational Biology}, 8\penalty0 (11):\penalty0 1--9,
  2012.
\newblock \doi{10.1371/journal.pcbi.1002785}.

\bibitem[Sol\'is-Lemus and An\'e(2016)]{SolisLemus2016}
Claudia Sol\'is-Lemus and C\'ecile An\'e.
\newblock Inferring phylogenetic networks with maximum pseudolikelihood under
  incomplete lineage sorting.
\newblock \emph{PLOS Genetics}, 12\penalty0 (3):\penalty0 e1005896, Mar 2016.
\newblock ISSN 1553-7404.
\newblock \doi{10.1371/journal.pgen.1005896}.

\bibitem[Wen and Nakhleh(2017)]{wen17}
Dingqiao Wen and Luay Nakhleh.
\newblock Coestimating reticulate phylogenies and gene trees from multilocus
  sequence data.
\newblock \emph{Systematic Biology}, 67\penalty0 (3):\penalty0 439--457, 2017.
\newblock \doi{10.1093/sysbio/syx085}.

\bibitem[Willson(2013)]{willson13}
Stephen~J. Willson.
\newblock Reconstruction of certain phylogenetic networks from their
  tree-average distances.
\newblock \emph{Bulletin of Mathematical Biology}, 75\penalty0 (10):\penalty0
  1840--1878, 2013.
\newblock \doi{10.1007/s11538-013-9872-z}.

\bibitem[Xu and An{\'e}(2023)]{2023XuAne_identifiability}
Jingcheng Xu and C{\'e}cile An{\'e}.
\newblock Identifiability of local and global features of phylogenetic networks
  from average distances.
\newblock \emph{Journal of Mathematical Biology}, 86\penalty0 (1):\penalty0 12,
  2023.
\newblock \doi{10.1007/s00285-022-01847-8}.

\bibitem[Yu and Nakhleh(2015)]{yu15}
Yun Yu and Luay Nakhleh.
\newblock A maximum pseudo-likelihood approach for phylogenetic networks.
\newblock \emph{BMC Genomics}, 16\penalty0 (S10):\penalty0 S10, 2015.
\newblock \doi{10.1186/1471-2164-16-s10-s10}.

\bibitem[Zhang et~al.(2017)Zhang, Ogilvie, Drummond, and Stadler]{zhang17}
Chi Zhang, Huw~A Ogilvie, Alexei~J Drummond, and Tanja Stadler.
\newblock Bayesian inference of species networks from multilocus sequence data.
\newblock \emph{Molecular Biology and Evolution}, 35\penalty0 (2):\penalty0
  504--517, 2017.
\newblock \doi{10.1093/molbev/msx307}.

\end{thebibliography}

\end{document}